\documentclass[11pt]{article}
\usepackage{amsmath, amsthm, amssymb}
\usepackage{graphicx,psfrag,epsf}
\usepackage{enumerate}
\usepackage{natbib}
\usepackage{url} 
\usepackage{rotating}
\usepackage{cleveref}
\usepackage{subcaption}
\usepackage{bbm}

\usepackage{booktabs}
\usepackage{tabularx}
\usepackage{threeparttable}
\usepackage{array}

\newtheorem{proposition}{Proposition}

\newcommand{\blind}{0}

\usepackage{tikz}
\usetikzlibrary{arrows.meta,positioning}

\definecolor{cartblue}{HTML}{2878B5}
\definecolor{mdfsred}{HTML}{D93838}
\definecolor{figuregray}{HTML}{777777}

\begin{document}

\def\spacingset#1{\renewcommand{\baselinestretch}%
{#1}\small\normalsize} \spacingset{1}


\if0\blind
{
  \title{\bf Policy Targeting with Binary Classification Trees: an Application to Rural Hospital Closures}
  \author{Hongying Li\thanks{The authors thank Kun Huang as a data contributor for this paper.} \\
    Department of Statistics, Ohio State University\\
    and \\
    Lei Bill Wang\thanks{Corresponding author. For inquiry about this paper or \texttt{targetree} package, kindly email \texttt{lei.wang@austin.utexas.edu}.} \\
    Department of Economics, University of Texas at Austin}

    \date{}
    
  \maketitle
} \fi

\if1\blind
{
  \bigskip
  \bigskip
  \bigskip
  \begin{center}
    {\LARGE\bf Title}
\end{center}
  \medskip
} \fi

\bigskip
\begin{abstract}
Empirical researchers often use binary classification trees to identify subgroups at risk of adverse outcomes. We compare two classification tree algorithms in this policy-targeting context: classical classification and regression trees (CART) and the maximizing-distance final-split approach (MDFS). We establish a theoretical setting in which MDFS identifies a high-risk subgroup while CART identifies none. Applied to rural hospital closure data, MDFS targets a high-closure-probability subgroup: for-profit hospitals with more than 63 days in accounts receivable. CART misses this subgroup. To facilitate further application of these classification tree algorithms, we provide \texttt{targetree} package in Python, R, and Stata.

\end{abstract}
\noindent%

\vspace{.5cm}

\noindent%
{\it Highlights:}
\begin{itemize}
    \item We compare the identification properties of classical classification and regression trees (CART) and the maximizing-distance final-split approach (MDFS).
    \item Applied to rural hospital closure data, MDFS finds a subgroup of rural hospitals at risk of imminent closure that CART misses.
    \item We provide \texttt{targetree} package in Python, R, and Stata.
\end{itemize}

\vspace{.5cm}

\noindent%
{\it Data availability:} The data that has been used by this paper is confidential. We include open access data in \texttt{targetree} for users to test the package.

\vfill

\newpage
\spacingset{1} 

\section{Introduction}

A substantial body of applied literature studies binary classification problems with the goal of identifying subgroups at high risk of adverse outcomes. For example, \cite{varian2014big} identify applicants at high risk of being denied mortgage loans; \cite{andini2018targeting} identify households at high risk of being financially constrained; and \cite{ash2025machine} identify municipalities at high risk of corruption.

A popular approach to such problems, adopted by all three studies, is a tree algorithm.\footnote{See \cite{wang2026policyoriented} for a list of 23 empirical papers that use tree algorithms for policy targeting. Some of the papers use random forest algorithms, which we discuss towards the end of this paper.} Tree algorithms recursively partition a population into subgroups according to observed characteristics. Researchers compare each terminal subgroup's average probability of the adverse outcome with a predetermined targeting threshold. A subgroup is targeted if its average probability exceeds that threshold. \Cref{fig tree-comparison} provides two examples of using tree algorithms for policy targeting.

\begin{figure}[p]
    \centering

    \begin{subfigure}[t]{\textwidth}
        \centering
        \includegraphics[
            width=\textwidth,
            height=0.38\textheight,
            keepaspectratio
        ]{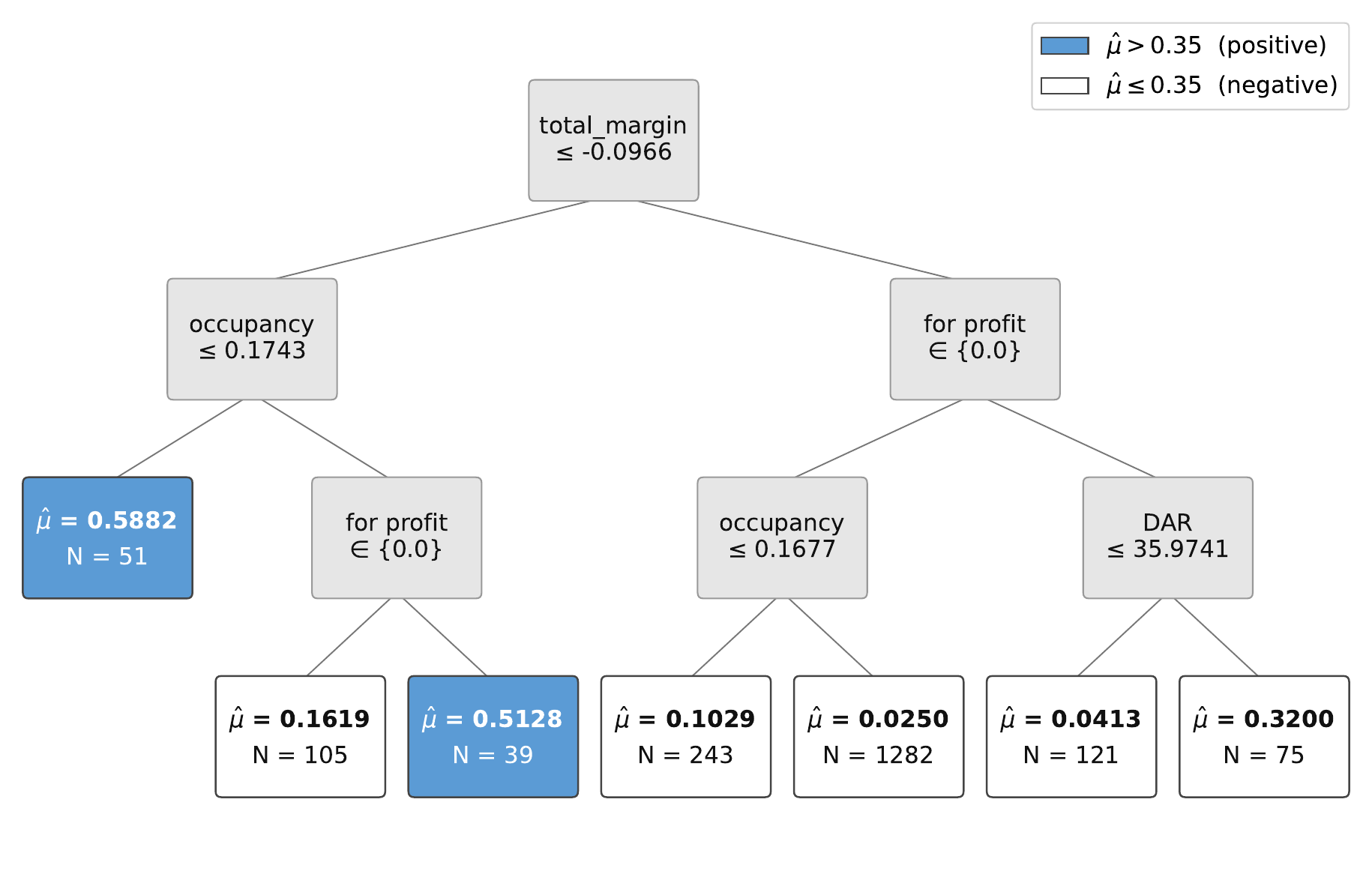}
        \caption{Targeting policy constructed using CART algorithm in \texttt{targetree}.}
        \label{fig cart}
    \end{subfigure}

    \vspace{1em}

    \begin{subfigure}[t]{\textwidth}
        \centering
        \includegraphics[
            width=\textwidth,
            height=0.38\textheight,
            keepaspectratio
        ]{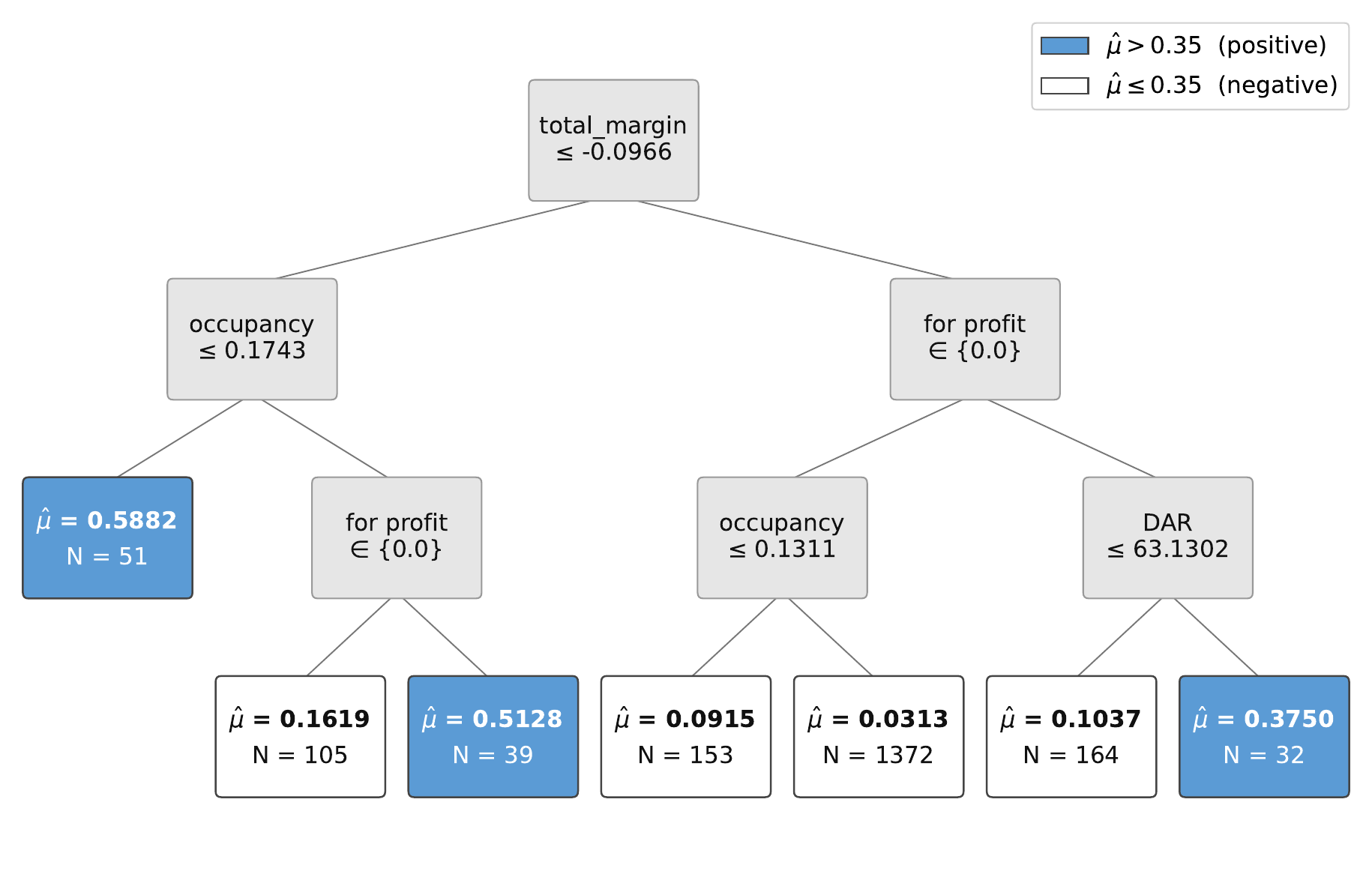}
        \caption{Targeting policy constructed using MDFS algorithm in \texttt{targetree}.}
        \label{fig mdfs}
    \end{subfigure}

    \caption{ \footnotesize
        Targeting policies constructed using CART and MDFS.
        Internal gray nodes display the splitting conditions. At each internal
        node, the left branch satisfies the displayed condition, whereas
        the right branch does not. Terminal nodes report the estimated
        subgroup closure rate, $\hat{\mu}$, and subgroup size, $N$.
        Blue terminal nodes are targeted because $\hat{\mu}>0.35$;
        white terminal nodes are not targeted.
    }
    \label{fig tree-comparison}
\end{figure}

This paper compares two binary classification tree algorithms used in this policy targeting context: classical Classification and Regression Tree (CART) and the Maximizing Distance Final Split approach (MDFS) from \cite{wang2026policyoriented}. In
\Cref{sec theoretical underpinning}, we establish a setting in which CART selects the same split regardless of the targeting threshold, whereas MDFS adaptively picks its split point depending on the predetermined threshold. As a result, MDFS identifies
a high-risk subgroup that CART misses. We formalize this comparison in an identification proposition.

To demonstrate the empirical value of our proposition, we study the rural hospital closure problem. Previous research documents the increasing rate of rural hospital closures and their consequences for rural communities \citep{kaufman2016rising,mills2024impact,mullens2024understanding}. In response, several studies develop conceptual and empirical frameworks for identifying rural hospitals at imminent risk of closure \citep{holmes2017predicting,balakrishnan2025predictors,malone2025updated}. 

\Cref{sec empirical application} applies CART and MDFS to a rural hospital closure data. As compared to CART, MDFS additionally identifies a subgroup with a probability of imminent closure above the threshold: for-profit hospitals with more than 63 days in accounts receivable. This aligns with our proposition. 

To facilitate replication and further applications, \Cref{sec conclusion} provides \texttt{targetree}, an accompanying package available in Python, R, and Stata. The package implements the class of binary classification trees developed by \citet{wang2026policyoriented}, including CART and MDFS, and visualizes their resulting targeting policies. 

\section{Illustrative identification proposition} 
Various tree algorithms use different split criterion functions to determine which covariate and at what values to split the population. This section sets up a simple example to compare the split criterion functions of CART and MDFS. 

For simplicity of this exposition, consider $(Y_i,X_{i1},X_{i2}) \overset{iid}{\sim} F$ across $i$. $Y_i \in \{0,1\}$ is the indicator for the adverse outcome. Both covariates $X_{i1}$ and $X_{i2}$ are uniformly distributed. This assumption is not too restrictive because any continuous variables can be converted into its uniformly distributed quantile when implementing a tree algorithm.


The policymaker has a predetermined threshold $c$ and would like to target subpopulations whose conditional probability of $Y_i = 1$ is greater than $c$. For example, in our empirical application, the policymaker aims to find those rural hospitals with greater than 35\% probability of closing down in two years, $c = 0.35$.

\label{sec theoretical underpinning}
\subsection{CART}
For a fixed covariate $X_{i1}$ and a fixed split point $s$, the population is divided into two subgroups, those with $X_{i1} \leq s$ and those with $X_{i1} > s$. Then, it computes the variance of $Y_i$ for each subgroup and sums the two variances weighted by their group sizes; see \Cref{eq split criterion function of CART}.
CART stores the weighted sum for different $s$ values of $X_{i1}$, and then repeats that with $X_{i2}$. Eventually, it picks the combination of covariate and $s$ that minimizes the following split criterion function 
\begin{align}
    \min_{j \in \{1,2\},s \in (0,1)} s \mu_L(s)(1 - \mu_L(s)) + (1-s) \mu_R(s)(1 - \mu_R(s)), \label{eq split criterion function of CART}
\end{align}
where $\mu_L(s) = \mathbbm{P}( Y_i = 1 \mid X_{ij} \leq s)$ and $\mu_R(s) = \mathbbm{P}( Y_i = 1 \mid X_{ij} > s)$, and $j \in \{1,2\}$. 

CART recursively splits the population using this procedure. The procedure ends when some prespecified condition is met. For example, the procedure permits a split only if each resulting child node contains at least 1\% of the population.


\subsection{Comparing CART and MDFS}
MDFS differs from CART only in its split criterion functions, which is
\begin{align}
    \max_{s \in (0,1)} s \lvert \mu_L(s) - c \rvert + (1-s) \lvert \mu_R(s) - c \rvert, \label{eq split criterion function of MDFS}
\end{align}
where $c$ is the predetermined threshold that the policymaker sets to define the high-risk subgroup.
To compare the split criterion functions \Cref{eq split criterion function of CART} and \Cref{eq split criterion function of MDFS}, we assume that (a) there is only one split and (b) drop $j \in \{1,2\}$ in \Cref{eq split criterion function of CART}. This is consistent with \cite{wang2026policyoriented} because (a) MDFS replaces only CART's \textit{final} layers' split criterion functions with \Cref{eq split criterion function of MDFS} and (b) uses the same split covariates as CART by design.

Without loss of generality, we assume that the covariate selected by CART for splitting the population is $X_{i1}$. Furthermore, we impose a parametric setup $\mathbbm{P}(Y_i = 1 \mid X_{i1}) = aX_{i1}$, where $1 > a > c > \frac{3}{4}a> 0$. 

Under this setup, the split point of CART is always 0.5 for any value of $a$ and $c$ (proof in Online Appendix \ref{proof for CART}). The $X_{i1} < 0.5$ subgroup has average $\mu_L(s^{CART}) = \frac{1}{4}a < c$; the $X_{i1} > 0.5$ subgroup has average $\mu_R(s^{CART}) = \frac{3}{4}a < c$. Neither subgroup is targeted. \Cref{fig cart mdfs theory} depicts this targeting policy.

\begin{figure}[!htbp]
\centering
\resizebox{\linewidth}{!}{%
\begin{tikzpicture}[
    font=\small,
    >=Latex,
    axis/.style={
        ->,
        semithick,
        draw=black!85
    },
    guide/.style={
        densely dashed,
        line width=0.8pt
    },
    response/.style={
        black!85,
        very thick
    },
    branch/.style={
        draw=figuregray,
        semithick
    },
    split node/.style={
        rectangle,
        rounded corners=2pt,
        draw=#1,
        line width=0.9pt,
        fill=white,
        minimum width=3cm,
        minimum height=0.85cm,
        align=center,
        inner sep=4pt
    },
    leaf node/.style={
        rectangle,
        rounded corners=2pt,
        draw=#1,
        line width=0.9pt,
        fill=white,
        text width=3.2cm,
        minimum height=1.05cm,
        align=center,
        inner sep=4pt
    }
]


\node[font=\bfseries] at (-0.55,5.55) {(a)};

\draw[axis]
    (0,0) -- (6,0)
    node[below=8pt] {$X_{i1}$};

\draw[axis] (0,0) -- (0,5.4);

\node[rotate=90] at (-1,2.7)
    {$\mathbbm{P}(Y_i=1\mid X_{i1})$};

\draw[response]
    (0,0) -- (5.35,5)
    node[
        pos=0.45,
        above,
        sloped
    ]
    {$\mathbbm{P}(Y_i=1\mid X_{i1})=aX_{i1}$};

\coordinate (cartpoint) at (2.675,2.5);

\draw[guide,cartblue]
    (2.675,0) -- (cartpoint);


\fill[cartblue]
    (cartpoint) circle (2.2pt);



\coordinate (cartpoint) at (0,3.75);

\fill[cartblue]
    (cartpoint) circle (2.2pt);

\coordinate (cartpoint) at (0,1.25);

\fill[cartblue]
    (cartpoint) circle (2.2pt);

\coordinate (mdfspoint) at (4.28*4.2/4,4.2);

\draw[guide,mdfsred]
    (4.28*4.2/4,0) -- (mdfspoint);

\draw[guide,black]
    (0,4.2) -- (mdfspoint);

\fill[mdfsred]
    (mdfspoint) circle (2.2pt);

\coordinate (mdfspoint) at (0,4.2);

\fill[black]
    (mdfspoint) circle (2.2pt);

\coordinate (mdfspoint) at (0,4.6);

\fill[mdfsred]
    (mdfspoint) circle (2.2pt);

\coordinate (mdfspoint) at (0,2.1);

\fill[mdfsred]
    (mdfspoint) circle (2.2pt);
    
\node[below] at (0,0) {$0$};

\node[
    below,
    text=cartblue,
    font=\bfseries
] at (2.675,0) {$0.5$};

\node[
    below,
    text=mdfsred,
    font=\bfseries
] at (4.28*4.2/4,0) {$c/a$};

\node[below] at (5.35,0) {$1$};

\node[left] at (0,0) {$0$};


\node[
    left,
    text=cartblue,
    font=\bfseries
] at (0,3.75) {$\frac{3}{4}a$};

\node[
    left,
    text=cartblue,
    font=\bfseries
] at (0,1.25) {$\frac{1}{4}a$};

\node[
    left,
    text=black,
    font=\bfseries
] at (0,4.2) {$c$};

\node[
    left,
    text=mdfsred,
    font=\bfseries
] at (0,4.6) {$\frac{a+c}{2}$};

\node[
    left,
    text=mdfsred,
    font=\bfseries
] at (0,2.1) {$\frac{c}{2}$};

\node[left] at (0,5) {$a$};

\node[
    above,
    text=cartblue,
    font=\bfseries
] at (2.675,-1.2) {$s^{CART}$};

\node[
    above,
    text=mdfsred,
    font=\bfseries
] at (4.28*4.2/4,-1.2) {$s^{MDFS}$};


\node[font=\bfseries] at (7,5.55) {(b)};

\node[
    text=cartblue,
    font=\bfseries\large,
    anchor=west
] at (6.55,4.55) {CART};

\node[
    split node=cartblue
] (cartroot) at (10.6,4.35)
    {$X_{i1}\leq 0.5$};

\node[
    leaf node=cartblue
] (cartleft) at (8.35,2.65)
    {$\mu_L=a/4<c$\\
     Not targeted};

\node[
    leaf node=cartblue
] (cartright) at (12.85,2.65)
    {$\mu_R=3a/4<c$\\
     Not targeted};

\draw[branch]
    (cartroot.south) -- (cartleft.north);

\draw[branch]
    (cartroot.south) -- (cartright.north);

\node[
    text=mdfsred,
    font=\bfseries\large,
    anchor=west
] at (6.55,1.05) {MDFS};

\node[
    split node=mdfsred
] (mdfsroot) at (10.6,0.85)
    {$X_{i1}\leq c/a$};

\node[
    leaf node=mdfsred
] (mdfsleft) at (8.35,-0.85)
    {$\mu_L=c/2<c$\\
     Not targeted};

\node[
    leaf node=mdfsred
] (mdfsright) at (12.85,-0.85)
    {$\mu_R=(a+c)/2>c$\\
     Targeted};

\draw[branch]
    (mdfsroot.south) -- (mdfsleft.north);

\draw[branch]
    (mdfsroot.south) -- (mdfsright.north);

\end{tikzpicture}%
}

\caption{
    Graphical comparison of the population splits selected by CART and MDFS. Panel~(a) shows the split points selected by the two methods on the horizontal axis in blue and red, respectively. It also shows the subgroup probabilities of CART and MDFS on the vertical axis in blue and red, respectively. Panel~(b) shows the
    resulting subgroup probabilities and targeting decisions.
}
\label{fig cart mdfs theory}
\end{figure}
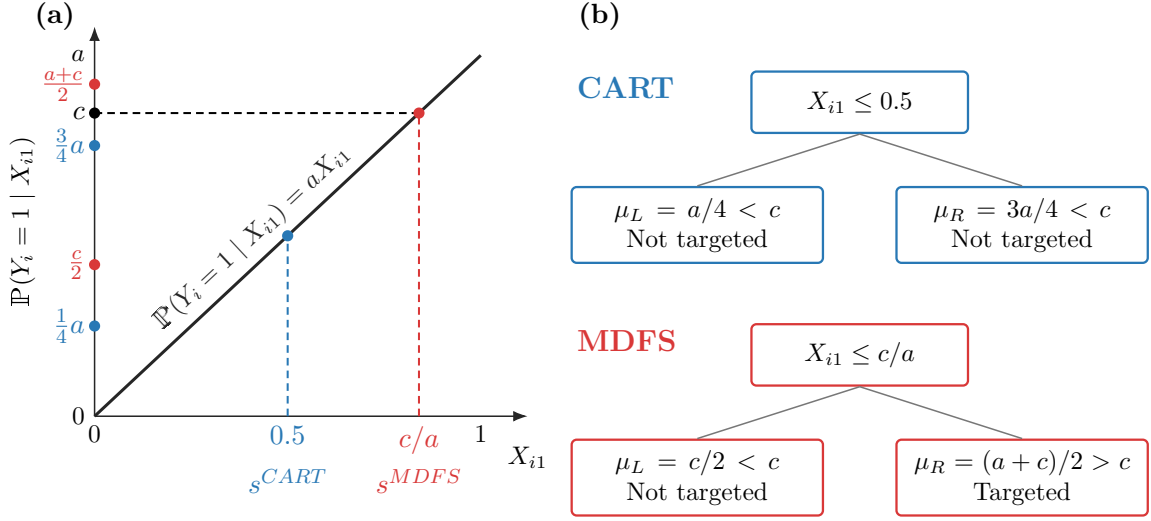

MDFS' split point depends on $a$ and $c$, $s^{MDFS} = \frac{c}{a}$ (proof in Online Appendix \ref{proof of proposition}). The split point is optimal as it perfectly divides the population into those with $\mathbbm{P}(Y_i = 1 \mid X_{i1}) \leq c$ and $\mathbbm{P}(Y_i = 1 \mid X_{i1}) > c$. The $X_{i1} \leq \frac{c}{a}$ subgroup has $\mu_L(s^{MDFS}) = \frac{c}{2} < c$, so this subgroup is not targeted; the $X_{i1} > \frac{c}{a}$ subgroup has $\mu_R(s^{MDFS}) = \frac{a+c}{2} > c$, so this subgroup is targeted. Hence, MDFS identifies an additional high-risk subgroup for targeting. \Cref{fig cart mdfs theory} depicts this targeting policy. We summarize this comparison between the two targeting policies as the following proposition.

\begin{proposition} \label{proposition parametric example}
Suppose that the chosen split covariate $X_{i1}\sim\operatorname{Uniform}(0,1)$ and
$\mathbbm{P}(Y_i=1\mid X_{i1})=aX_{i1}$, where $1\geq a>c>\frac{3}{4}a>0$. MDFS identifies the optimal split point $s^{\mathrm{MDFS}}=\frac{c}{a}$.
Consequently, the $X_{i1} > \frac{c}{a}$
subgroup produced by MDFS has
$\mu_R(s^{MDFS}) > c$ and is therefore targeted. CART misses this subgroup.
\end{proposition}
\Cref{proposition parametric example} is an intuitive illustration of Theorem 4.4 and Remark 4.7 in \cite{wang2026policyoriented}. Theorem 4.4 generalizes the result that MDFS identifies the optimal split point to a large class of nonparametric $\mathbbm{P}(Y_i = 1 \mid X_{i1})$, which encompass any monotonic differentiable function that has a single crossing over $c$.

\section{Application to rural hospital closure} \label{sec empirical application}
\subsection{Data}
We merge the proprietary RAND hospital data with the open-access Area
Health Resources Files and rural hospital closure records. Online Appendix \ref{sec data processing}  details how we construct the final sample. The sample contains hospital identifiers, geographic
information, closure years if closed within sample period, and measures of hospital finances,
organization, government reimbursement, and local market conditions. It covers 2001--2022 and includes 1,916 hospitals, of which 153 closed during the sample period. Each hospital is observed exactly one time in the sample. That one observation of a hospital contains covariates in year $t$ and the hospital open/closure status in year $t+2$. \Cref{table summary_predictors} in Online Appendix \ref{sec data processing} reports
summary statistics for the sample's covariates.

\subsection{Results}
We feed the merged data to the CART and MDFS algorithms in \texttt{targetree}. We set $c = 0.35$, depth of 3, and specify that terminal subgroup has node size greater than 30.

\Cref{fig cart} depicts the targeting policies constructed by CART. CART targets two subgroups. The first subgroup are those with total margin below -9.66\% and occupancy below 17.43\%. The second subgroup are those with total margin below -9.66\%, occupancy higher than 17.43\%, are for-profit. The policy suggests that among all low total margin hospital, non-profit and for-profit hospitals should be differentially targeted. \textit{All} for-profit hospitals with total margin below -9.66\% should be targeted, whereas among non-profit hospitals with total margin below -9.66\%, \textit{only} those with occupancy lower than 17.43\% should be targeted.

\Cref{fig mdfs} depicts the targeting policies constructed by MDFS. MDFS targets three subgroups. The first two subgroups are identical to CART's targeting subgroups. Consistent with the proposition, MDFS additionally targets those hospitals with total margin greater than -9.66\%, are for-profit, and has more than 63.13 days in accounts receivable (DAR). 
The additional targeting subgroup indicates that  not only does the eventual profitability (total margin) matter; the timeline (DAR) for a for-profit hospital to become profitable also matters.

\section{\texttt{targetree} package} \label{sec conclusion}
To facilitate further applications of tree algorithms to policy targeting, we provide \texttt{targetree}, a package that is available in Python, R, and Stata. As shown the empirical application, \texttt{targetree} can be used for two purposes: first, constructing interpretable targeting policies, each targeted subgroups can be clearly interpreted as an intersection of a small number of split rules; second, it can highlight the policy relevance of prior causal/associative findings, for example, \cite{holmes2017predicting} uses cash flow as one of the key variables in the rural hospital financial distress model, MDFS' targeting policies show that cash flow varible (DAR) is critical for accurately targeting for-profit hospitals.

In addition to CART and MDFS, \texttt{targetree} implements two extension tree algorithms, one of which leverages the random forest to potentially improve the targeting accuracy. More details of these two extensions can be found in Online Appendix \ref{sec two extensions}.

\bibliographystyle{chicago}

\bibliography{Bibliography}

\newpage
\appendix
\begin{center}
    \Large
    \textbf{Online Appendix}
\end{center}

\section{Proofs} 
\subsection{Proof of $s^{CART} = 0.5$}\label{proof for CART}
\begin{proof}
    Substitute $\mathbbm{P}(Y_i = 1 \mid X_{i1}) = aX_{i1}$ into \Cref{eq split criterion function of CART}, we minimize
\begin{align*}
    s\frac{as}{2}(1 - \frac{as}{2}) + (1-s) \frac{as+a}{2} (1 - \frac{as+a}{2}).
\end{align*}
The first order condition of this minimization problem is $\frac{a^2}{4}(2s-1) = 0$, which outputs a stationary point at $s^{CART} = 0.5$ and the second order derivative is $\frac{a^2}{2} > 0$. Hence, $s^{CART} = 0.5$ is the argmin of \Cref{eq split criterion function of CART}. 
\end{proof}
\subsection{Proof of $s^{MDFS} = \frac{c}{a}$} \label{proof of proposition}
\begin{proof}
Substitute $\mathbbm{P}(Y_i = 1 \mid X_{i1}) = aX_{i1}$ into \Cref{eq split criterion function of MDFS}, we maximize
\begin{align*}
    s \left\lvert \frac{as}{2} - c\right\rvert + (1-s) \left\lvert \frac{as+a}{2} - c \right\rvert.
\end{align*}
Because $as/2 < a/2 < \frac{3}{4}a <c$ for every $s\in(0,1)$,
the sign of the first term inside the absolute value does not change.
The sign of the second term changes at $s=2c/a-1$. Thus, the MDFS
objective, denoted here as $Q(s)$, can be written as
\[
Q(s)=
\begin{cases}
c-\dfrac{a}{2},
    & 0<s\leq 2sa/c-1, \\[6pt]
\dfrac{a}{2}-c+2cs-as^2,
    & 2a/c-1<s<1.
\end{cases}
\]
On the second region, $Q'(s)=2c-2as$ and $Q''(s)=-2a<0$.
Its unique maximum is therefore $s=c/a$. Moreover,
$c/a>2c/a-1$ because $c/a<1$, so this point lies in the
second region. Finally,
\[
Q\left(\frac{c}{a}\right)
-\left(c-\frac{a}{2}\right)
=\frac{(a-c)^2}{a}>0.
\]
Hence, the global maximizer is uniquely given by
$s^{\mathrm{MDFS}}=c/a$.
\end{proof}

\section{Data processing} \label{sec data processing}
We impute missing numerical covariates first using the hospital-specific
median across available years, then the corresponding state-year median
when a hospital-specific median is unavailable, and finally the
full-sample median. We apply the analogous procedure to categorical
covariates using the mode.

The data cover 2001--2022 and include 153 rural hospitals that closed
during the sample period. We construct an analytic sample containing one
observation per hospital. For each closed hospital, we retain its
observation two years before closure. For each hospital that remains open
throughout the sample period, we randomly select one year of observation.
This procedure reduces the 37,703 open hospital-year observations to
1,763 hospital-level observations. The resulting analytic sample contains 1,916 hospitals, of
which 153 closed during the sample period. \Cref{table summary_predictors} reports
summary statistics for the sample's covariates.

\begin{table}[!htbp]
\centering
\caption{Summary Statistics}
\label{table summary_predictors}
\begin{threeparttable}
\setlength{\tabcolsep}{4pt}
\renewcommand{\arraystretch}{1.15}

\begin{tabularx}{\textwidth}{
    >{\raggedright\arraybackslash}p{3.5cm}
    >{\raggedleft\arraybackslash}p{1.5cm}
    >{\raggedleft\arraybackslash}p{1.5cm}
    >{\raggedleft\arraybackslash}p{1.5cm}
    >{\raggedright\arraybackslash}X
}
\toprule
Variable
    & Mean
    & Median
    & Std. dev.
    & Definition and notes \\
\midrule

\multicolumn{5}{l}{\textit{Financial performance}} \\[2pt]

Days in accounts receivable (DAR)
    & 55.160 & 46.830 & 64.797
    & Average days it takes to collect payment for services rendered \\

Inpatient margin
    & 0.838 & 0.483 & 2.220
    & Inpatient revenue/Inpatient cost - 1 \\

Outpatient margin
    & 2.253 & 1.588 & 2.346
    & \\


Total margin
    & 0.020 & 0.032 & 0.223
    &  \\

Negative equity
    & 0.173 & 0.000 & 0.379
    & Equity = Total asset - total liability $<$ 0 \\

Equity decline
    & 0.271 & 0.000 & 0.445
    & Equity decline $>$ 20\% over two years  \\

\addlinespace[4pt]
\multicolumn{5}{l}{\textit{Government reimbursement}} \\[2pt]

Medicare share
    & 0.590 & 0.600 & 0.176
    & Medicare as percentage of service \\

Medicaid service
    & 0.264 & 0.000 & 0.441
    & Provide Medicaid inpatient service or no? \\

\addlinespace[4pt]
\multicolumn{5}{l}{\textit{Organizational characteristics}} \\[2pt]

Beds
    & 45.589 & 25.000 & 51.991
    &  \\

For-profit status
    & 0.128 & 0.000 & 0.335
    &  \\

Occupancy
    & 0.356 & 0.334 & 0.188
    & Percentage of beds filled by patients \\

Critical Access Hospital (CAH)
    & 0.579 & 1.000 & 0.494
    &  \\

\addlinespace[4pt]
\multicolumn{5}{l}{\textit{Market characteristics}} \\[2pt]

Population (in thousands)
    & 25.289 & 19.116 & 22.070
    & County population \\

Proportion in poverty
    & 0.151 & 0.142 & 0.056
    &  County poverty rate\\

Hospitals in county
    & 2.048 & 1.000 & 4.594
    & Number of neighboring hospitals \\

\bottomrule
\end{tabularx}

\begin{tablenotes}[flushleft]
\footnotesize
\item \textit{Notes:}
The table reports summary statistics for the covariates used to
construct the CART and MDFS.
All statistics are calculated using the final analytic sample after data
cleaning, imputation, and sample selection. Most of the variables follow the definitions in Table 1 of \cite{holmes2017predicting}. 
\end{tablenotes}
\end{threeparttable}
\end{table}

\section{\texttt{targetree}}\label{sec two extensions}
We provide \texttt{targetree} package in Python, R, and Stata. 
\begin{itemize}
    \item Python: \texttt{https://github.com/lhy-0594/targetree-python/tree/main}
    \item R: \texttt{https://github.com/lhy-0594/targetree-r}
    \item Stata: \texttt{https://github.com/lhy-0594/targetree-stata}
\end{itemize}
The package implements the class of binary classification tree algorithms in \cite{wang2026policyoriented}, including CART and MDFS.

\subsection{PFS}
\texttt{targetree} includes an additional tree algorithm named penalized final split approach (PFS) whose split criterion function is a convex combination of \Cref{eq split criterion function of CART} and \Cref{eq split criterion function of MDFS}. 
\begin{align*}
    \min_{s \in (0,1)} (1 - \lambda) \Cref{eq split criterion function of CART} - \lambda \Cref{eq split criterion function of MDFS} 
\end{align*}
where $\lambda \in [0,1]$. When $\lambda = 0$, PFS is equivalent to CART; when $\lambda = 1$, PFS is equivalent to MDFS. PFS outperforms CART under a more general condition than MDFS, but it does not attain the optimal split point like MDFS, see Theorem 5.1 in \cite{wang2026policyoriented}.

\subsection{Knowledge distillation}
Tree algorithms tend to perform well in an ensemble. Instead of building one tree, usually many trees are constructed simultaneously using the same data. Then, the trees would each make a prediction of the binary decision, and by majority voting, such a ``random forest'' algorithm makes a final prediction. Though the random forest algorithm has been proven to be effective for prediction, it does not output a simple targeting policy that can be communicated concisely. Some works have proposed using random forest as a teacher model to extract knowledge from the data, and then use CART as a student model to distill the complex knowledge that the teacher model stores \citep{dao2021knowledge}. Such a method is called ``knowledge distillation''. We include two knowledge distillation algorithms in \texttt{targetree} Python package. Both use random forest as the teacher model. One of them uses CART as the student model; the other uses MDFS as the student model.

\end{document}